\documentclass{article}
\usepackage{arxiv}

\usepackage{balance} 

\usepackage{amsmath}
\usepackage{amsfonts}
\usepackage{pgfplots}
\usepackage{xcolor} 
\pgfplotsset{compat=newest}
\usepackage{pgfplots}
\pgfplotsset{compat=1.17}
\usepackage{pgfplotstable}
\usepackage{tikz}
\usepackage{svg}
\usepackage{newfloat}
\usepackage{tabularx}
\usepackage{booktabs} 
\usepackage{amsmath}  
\usepackage{array}
\usepackage[normalem]{ulem}
\usepackage{newfloat}
\usepackage{amsthm}
\usepackage{amssymb}

\title{FLIGHT: Facility Location Integrating Generalized, Holistic Theory of Welfare}

\author{
    Avyukta Manjunatha Vummintala \\ 
    IIIT Hyderabad \\ 
    \texttt{avyukta.v@research.iiit.ac.in} 
    \And 
    Shivam Gupta \\ 
    IIT Ropar \\ 
    \texttt{shivam.20csz0004@iitrpr.ac.in} 
    \AND
    Shweta Jain \\ 
    IIT Ropar \\ 
    \texttt{shwetajain@iitrpr.ac.in} 
    \And 
    Sujit Gujar \\ 
    IIIT Hyderabad \\ 
    \texttt{sujit.gujar@iiit.ac.in} 
}

\newcommand{\BibTeX}{\rm B\kern-.05em{\sc i\kern-.025em b}\kern-.08em\TeX}

\newtheorem{theorem}{Theorem}
\newtheorem{definition}{Definition}

\newtheorem{lemma}{Lemma}

\newcommand{\ourframework}{FLIGHT}
\definecolor{LightViolet}{RGB}{143,0,255}  

\begin{document}


\pagestyle{fancy}
\fancyhead{}


\maketitle


\begin{abstract}
The Facility Location Problem (FLP) is a well-studied optimization problem with applications in many real-world scenarios. Past literature has explored the solutions from different perspectives to tackle FLPs. These include investigating FLPs under objective functions such as utilitarian, egalitarian, Nash welfare, etc. We propose a unified framework, FLIGHT, to accommodate a broad class of welfare notions. The framework undergoes rigorous theoretical analysis, and we prove some structural properties of the solution to FLP. Additionally, we provide approximation bounds, which (under certain assumptions) provide insight into an interesting fact-- as the number of agents arbitrarily increases, the choice of welfare notion is irrelevant. Furthermore, the paper examines a scenario in which the agents are independently and identically distributed (i.i.d.) according to a given probability distribution. In this setting, we derive results concerning the optimal estimator of the welfare and establish an asymptotic result for welfare functions.
\end{abstract}

\keywords{Facility Location; Welfare functions}


\section{Introduction}
\label{sec:intro}

The most commonly studied \emph{Facility Location Problem} (FLP) considers the problem of placing a facility on a line segment (typically normalized as $[0,1]$). 
Here, agents derive certain utilities from this facility. The goal of a planner is to ensure the welfare of the agents who use this facility is maximized. 
Traditional approaches typically rely on predefined welfare functions such as \emph{utilitarian welfare} ~\cite{feigenbaum2017approximately}, which maximizes the total welfare (e.g., travel distance), or \emph{egalitarian welfare}~\cite{Procaccia_Wajc_Zhang_2018}, which maximizes the minimum welfare.

While these approaches are suitable in many scenarios, they may not be sufficient to capture the complexity and nuances of real-world applications, especially when the relationship between agents and the facility involves non-linear or application-specific factors. Factors such as varying environmental conditions and resource constraints can introduce non-linearities that complicate the optimization process. In many cases, using simple distance-based models can lead to suboptimal solutions that fail to reflect the true welfare of the system.  Moreover, with the rise of \emph{artificial intelligence} (AI) and \emph{machine learning} (ML) techniques, there is a growing trend toward \emph{learning} welfare functions from data rather than relying on predefined or assumed models. Recent advancements in machine learning have shown that welfare functions can be inferred directly from historical data, allowing for a more flexible and context-aware approaches ~\cite{pardeshi2024learning,zadimoghaddam2012efficiently,balcan2014learning}. This shift towards data-driven models underscores the importance of a generalized framework that can accommodate learned welfare functions, making it adaptable to changing environments. 
Hence, a \emph{generalized framework} is necessary to accommodate various welfare functions that can adapt to these complexities.

One approach to generalizing welfare functions is through the use of \emph{\( p \)-mean functions}~\cite{bourbaki1987topological,gupta2022l_p}, which provide a continuous spectrum of solutions ranging from utilitarian welfare (when \( p = 1 \)) to egalitarian welfare (when \( p = \infty \)). \( p \)-mean functions allow more control over the system's balance between efficiency and fairness.                                                        
\( p \)-mean functions, while powerful, may still not fully capture the diversity of welfare considerations present in real-world applications. In this paper, we seek to go beyond \( p \)-mean functions by introducing a \emph{generalized framework} that allows for the inclusion of a wide variety of welfare functions, including but not limited to \( p \)-mean functions. Our \emph{generalized welfare framework} -- \ourframework, is designed to handle welfare functions that are learned from data, or defined based on specific application needs.  \ourframework\ framework is flexible enough to incorporate traditional welfare functions, such as utilitarian and egalitarian welfare, and more complex welfare functions that arise in modern applications as long as the welfare function is non-increasing for each agent from its location.

We begin by establishing several key properties of generalized welfare functions, such as \emph{concavity} and \emph{location invariance}, which are essential for ensuring tractable optimization. Our goal is to study the structural properties of the solution to FLP modelled in \ourframework. First, we explore more specific properties of these welfare functions. {Next, for  a concave positive welfare function $\alpha$ and  an arbitrary welfare function $\beta$, we provide bound on optimal welfare achieved by $\alpha$ with respect to the welfare achieved by $\beta$}

Our results show that a class of generalized welfare functions can approximate others with a constant approximation ratio, ensuring that our framework remains efficient. 
Often, a practitioner might be more interested in expected welfare than exact welfare for every instance. Towards this,
we investigate probabilistic versions of the facility location problem in which the agents are independently and identically distributed (i.i.d.) according to a given probability distribution. 
We derive results concerning the optimal estimator of the welfare and provide an asymptotic property of welfare functions.
In summary, our contributions are as follows:

\subsection{Our Contributions}
\label{sec:contributions}

\begin{enumerate}
    
    
    \item We propose a unified framework \textit{FLIGHT} that is capable of accommodating classical welfare functions, including utilitarian, egalitarian, and Nash welfare functions (Section \ref{sec:flight_frmwrk})
    
    \item {Under the concavity assumption of utility function,} we derive a series of theoretical results concerning the structural properties of generalized welfare functions. (Section ~\ref{sec:alpha_welfare}) Specifically, we prove:
    \begin{itemize}
        \item \textbf{Theorem 1:} Concavity of the welfare function,
        \item \textbf{Theorem 2:} Location invariance of the welfare function,
        \item \textbf{Theorem 3:} Behavior under agent shifts, and
        \item \textbf{Theorem 4:} Maximum shift property.
    \end{itemize}
    We then explore more specialized properties under stronger assumptions, including: \
    \begin{itemize}
        \item \textbf{Theorem 5:} Constant approximation bound for concave and positive utility functions,
        \item \textbf{Theorem 8:} Bounding distance between the peaks based upon the agent location profile.
    \end{itemize}
    \item In Section ~\ref{sec:prob_alpha_welfare}, we extend the analysis to probabilistic versions of the facility location problem. We establish estimation bounds (\textbf{Theorem 9} and \textbf{Theorem 10}) and derive asymptotic results, including \textbf{Theorem 11}, which provides an asymptotic property of welfare functions. 
    
\end{enumerate}
\section{Related Work}

The facility location problem (FLP) has a long and rich history, with its origins tracing back to $17^{\text{th}}$-century mathematicians like Pierre de Fermat and Evangelista Torricelli, who studied geometric optimization problems involving the positioning of points to minimize distances to a given set of locations, known as the Fermat-Weber problem~\cite{drezner1995facility}. This early work laid the foundation for modern FLP. The field saw significant growth after World War II, spurred by advances in operations research, as facility location became crucial for industrial planning, supply chains, and logistics~\cite{koopmans1957assignment}. During this period, figures such as Harold Kuhn formalized mathematical models that enabled the practical application of FLP to real-world challenges, ranging from public service placement to telecommunications infrastructure~\cite{weber1909standort}.
In modern times, the facility location problem has found broad applications in diverse fields such as operations research, computer science, and electronics. With the rise of data-driven decision-making, facility location models are now applied in cloud computing infrastructure, data centers, network design, and even in the placement of sensors in wireless networks ~\cite{sl1979algorithmic}. The continued relevance of facility location models underscores their versatility in addressing problems that require optimal resource allocation and spatial planning.

\subsection{General Facility Location}
The general facility location problem has been widely studied across various fields due to its applications in logistics, urban planning, and operations research. 
A general overview of the results and variants of FLPs can be found in ~\cite{farahani2009facility,melo2009facility, chan2021mechanism}. Several variants of the FLP have been studied, such as obnoxious facility location~\cite{tamir1991obnoxious} and capacitated facility location ~\cite{wu2006capacitated}. Online FLPs are also studied where the agents arrive in an online fashion and a set of facilities is maintained~\cite{meyerson2001online,fotakis2008competitive}.
Additionally, \cite{Leitner2014} examines the polytope associated with the asymmetric version of the facility location problem. 
\cite{fotakis2013strategyproof} also study facility location with concave welfare functions. However, their focus is on designing algorithms with a constant approximation ratio, whereas our work investigates the structural properties of such a system. \cite{snyder2006facility} considers a probabilistic view of FLPs. This is relevant as we also perform a probabilistic analysis.

\subsection{Facility Location on a Line, Fairness, and Strategyproofness}
The facility location problem on a line, where both agents and facilities are confined to a linear domain, has garnered significant attention for its simplicity and traceability. \cite{procaccia2013approximate, Procaccia_Wajc_Zhang_2018} provide approximation guarantees to deterministic and randomized mechanisms that try to minimize total cost while maintaining strategyproofness to ensure no agent can manipulate the outcome. These works highlight the need for welfare functions that incorporate fairness, and our framework addresses this requirement.

Recent work on fairness in facility location problems has become increasingly relevant as a growing emphasis has been placed on equitable distribution across agents ~\cite{Li_Li_Chan_2024,chen2020strategyproof,lam2021balancing}. \cite{moulin2004fair} introduced the Nash welfare function, establishing its foundational role in welfare economics. \cite{lam2023nash} further highlight its application in facility location, demonstrating that the Nash welfare function effectively balances fairness and efficiency.
 This is particularly important for our generalized welfare framework, which aims to extend beyond specific functions like Nash welfare. Furthermore, \cite{chen2020strategyproof} introduce algorithms for 2-facility location that ensure envy-freeness, reinforcing the importance of fairness in our work.  \cite{aziz2023proportional} examine the problem of \textit{proportional fairness} in obnoxious facility location, where facilities are undesirable to agents and fairness becomes a key concern. \cite{wang2024positive} introduce the concept of \textit{positive intra-group externalities} in facility location, focusing on how intra-group dynamics affect utility and strategyproof mechanisms \cite{wang2024positive}. 

\subsection{Welfare Functions and $p$-mean Functions}
Welfare functions have long been central to decision-making and resource allocation in facility location. 
\cite{balcan2014learning} presents a method for learning welfare functions from revealed preferences, which is critical as our generalized framework aims to accommodate complex and dynamically evolving welfare functions. ~\cite{pardeshi2024learning} explores the theoretical front of learning welfares or preferences through the context of generalization bounds. In the context of Nash welfare, \cite{Kurokawa2016} demonstrates its use in allocation problems, reinforcing the importance of designing flexible welfare functions that balance fairness and efficiency.

Researchers have also explored generalizations of utilitarian and egalitarian welfare through $p$-mean functions ~\cite{gupta2022l_p}, which can be viewed as a parameterized family of welfare functions where varying the parameter $p$ adjusts the balance between fairness and efficiency ~\cite{daskin2015p}. For instance, $p = 1$ corresponds to utilitarian welfare, $p = \infty$ corresponds to egalitarian welfare, and intermediate values of $p$ provide trade-offs between these extremes. Our work builds on these concepts by integrating $p$-mean functions into a broader framework for generalized welfare functions.

\cite{barman2024compatibility} and \cite{lam2023nash} contribute to the growing body of work on Nash welfare, focusing on balancing fairness and efficiency. Our framework expands on these ideas by allowing for general welfare functions that can capture more complex and non-linear utility structures, as noted by ~\cite{Drezner2001} in facility location. The increasing need for \textit{learned} generalized welfare functions \cite{pardeshi2024learning,balcan2014learning,zadimoghaddam2012efficiently} to accommodate engineering applications and other real-world complexities further motivates our research. In the next section, we will introduce the formal problem setup and explain the notations.

\section{Preliminaries}
\label{sec:prelim}

\subsection{Facility Location Problem Setup}
We consider a scenario in which a set of \( n \) agents, denoted by \( N = \{1, \dots, n\} \), are positioned along the interval\footnote{Note that the [0,1] domain can be extended and translated to be any closed interval.} \([0,1]\). Each agent \( i \in N \) is located at a specific point \( x_i \in [0,1] \), and the collective set of agent locations is represented by the vector \( \mathbf{x} = (x_1, \dots, x_n) \). Without loss of generality (w.l.o.g.), we assume that the agent positions are ordered such that \( x_1 \leq x_2 \leq \dots \leq x_n \).

The social planner's problem is placing a single facility that serves these agents. Let the mechanism of this mapping be \( f : [0,1]^n \rightarrow [0,1] \), which takes the vector of agent locations \( \mathbf{x} \) as input and returns a location \( y \in [0,1] \) for the facility. 
{For a facility at location \(y\), agent \(i\) need to travel \(\mid y-x_i\mid\).  Thus, \(\mid y-x_i\mid\) indicates the cost to it or in some contexts, \(1 - \mid y-x_i\mid\) indicates the utility to agent $i$. The most prominently studied welfare functions are computed
as follows.}
\begin{definition}[Utilitarian Welfare]
    For the agents located at \(\mathbf{x}\) and the facility located at \(y\), the \emph{Utilitarian Welfare} is $$W_{\text{Utilitarian}}(y,\mathbf{x}) = \sum_i (1 - \mid y - x_i \mid) $$
\end{definition}

\begin{definition}[Egalitarian Welfare]
    For the agents located at \(\mathbf{x}\) and the facility located at \(y\), the \emph{Egalitarian Welfare} is $$W_{\text{Egalitarian}}(y,\mathbf{x}) = \min_i (1 - \mid y - x_i \mid )$$
\end{definition}

\begin{definition}[Nash  Welfare]
    For the agents located at \(\mathbf{x}\) and the facility located at \(y\), the \emph{Nash Welfare} is $$W_{\text{Nash}}(y,\mathbf{x}) = \prod_i \left(1 -  \mid y - x_i \mid \right) $$
\end{definition}

Typically, the social planer aims to place the facility at a location $y$ that maximizes $W_{\text{Utilitarian}}$ or $W_{\text{Nash}}$ or $W_{\text{Egalitarian}}$. There are closed-form solutions for Utilitarianism and Egalitarianism. 

\subsection{Key Important Mechanisms}
\label{sec:key}
In facility location problems (FLP), different welfare optimization criteria lead to distinct placement strategies for the facility.

The solution that maximizes \textit{utilitarian welfare}—defined as the total sum of utilities—is the median of the agent locations. Formally, this position is given by \( x_{\lfloor n/2 \rfloor} \), where \( n \) represents the total number of agents. This placement has the additional advantage of being \textit{strategyproof}, meaning agents cannot benefit from misreporting their locations.

In contrast, the solution that maximizes \textit{egalitarian welfare} (focused on maximizing the minimum utility for any agent)—is the midpoint between the extreme agents. This solution can be expressed as \( \frac{x_1 + x_n}{2} \), where \( x_1 \) and \( x_n \) represent the positions of the agents at the two extremes.

Finally, the solution that maximizes \textit{Nash welfare}—a balance between utilitarian and egalitarian objectives—is more complex. The Nash welfare function is the product of individual utilities, and finding its maximization in FLP is known to be both difficult to compute and interpret in practice \cite{moulin2004fair}.

While these solutions maximize different welfare objectives, many interesting properties emerge from their comparative analysis. However, these properties have traditionally been studied separately for each welfare function. This paper proposes a unifying framework that allows for the study of these properties in a more general, abstract manner.

As stated previously (Section \ref{sec:intro}), there are scenarios where one must go beyond the three classical welfare functions. Rather than developing new solutions for each emerging welfare criterion, a more holistic approach can be adopted. Specifically, we want to study facility location as an abstract problem, agnostic to the specific welfare function, by focusing on common properties shared by many of these functions. One such approach involves the use of $p$-mean functions, which we explain in the next section.

\subsection{$p$-mean Welfare Functions}

Consider the facility location problem where the \( L_p \)-norm is used as the distance metric between agents and the facility. 
The solution \( y_{Pmean} \) to the facility location problem under the \( L_p \)-norm is defined as the facility location that minimizes the \( p \)-mean distance to all agents, given by:

\begin{equation}
    y_{Pmean} = \arg \min_{y \in [0,1]} \left( \sum_{i \in N} |y - x_i|^p \right)^{1/p}
\end{equation}

Since the \( p \)-th root is a monotonically increasing function, we can simplify the optimization problem to:

\begin{equation}
\label{eq:p_mean_opt}
    y_{Pmean} = \arg \min_{y \in [0,1]} \sum_{i \in N} |y - x_i|^p
\end{equation}

{Next section proposes a more general framework, \emph{\ourframework} -- \textbf{Facility Location Integrating Generalized, Holistic Theory of Welfare}.  \ourframework\ generalizes the concept of welfare functions and provides a unified approach to solving facility location problems. }

\section{A Unified Perspective}
\label{sec:flight_frmwrk}
 
We propose \ourframework\  and demonstrate how all well-studied welfare functions, including $p$-mean functions, can be incorporated into it. 
We show that the Nash Welfare function can also be integrated within the FLIGHT framework, thereby highlighting the versatility and generality of our approach in encompassing a wide range of welfare functions.

\subsection{FLIGHT Framework}

Utility for an agent at location $x_i$ when the facility is located at $y$ is a function of $y-x_i$. Let
the utility function for each agent be \( \alpha: \mathbb{R} \rightarrow \mathbb{R} \). Specifically, \( \alpha \) takes the distance from the facility as input and returns the corresponding utility for the agent as output. The function \( \alpha \) encapsulates how the agent’s utility diminishes with increasing distance from the facility\footnote{Note: $\alpha$ could be an asymmetric functions as well, meaning it lacks symmetry about the y-axis.}. 







Next, we define the \textbf{total welfare} \( W_{\alpha}(y, \mathbf{x}) \) as the aggregate of individual utilities across all agents. Formally, it is expressed as:

\[
W_{\alpha}(y, \mathbf{x}) = \sum_{i \in N} \alpha(y - x_i)
\]

where \( y \in [0,1] \) represents the location of the facility, and \( \mathbf{x} = (x_1, \dots, x_n) \) denotes the vector of agent locations.

Given this total welfare function, the social planner's goal is to determine a location that maximizes global welfare. We denote it as \( P_{\alpha}(\mathbf{x}) \). Formally, this can be expressed as:

\[
P_{\alpha}(\mathbf{x}) = \arg \max_{y \in [0,1]} W_{\alpha}(y, \mathbf{x})
\]

For Utilitarian welfare, as stated in Sec.~\ref{sec:key},  $P_{\alpha}(x)=x_{\frac{n}{2}}$ and for Egalitarian welfare, $P_{\alpha}(x)=\frac{x_1+x_n}{2}$.
{In the next section, we show that $p$-mean welfare functions are special cases of our framework.}

\subsection{Incorporating $p$-mean Welfare Functions into Our Framework}

In this section, we demonstrate that $p$-\textit{Mean utility functions} are fully accommodated by our framework. 
Since \textit{utilitarian welfare} and \textit{egalitarian welfare} are special cases of $p$-mean utility functions, which naturally fit within our general framework.


\subsubsection{Generalizing to $p$-Mean Utility Functions}
The idea is to align optimization from Eq~\ref{eq:p_mean_opt} with \ourframework, we can express it as a maximization problem as:

\begin{equation}
    y_{Pmean} = \arg \max_{y \in [0,1]} - \sum_{i \in N} |y - x_i|^p
\end{equation}


To incorporate $p$-mean utility functions into our framework, we define a utility function \( \alpha: \mathbb{R} \rightarrow \mathbb{R} \), where $\alpha(x) = -|x|^p$.
Using this definition, the total welfare function \( W_\alpha(y, \mathbf{x}) \) becomes:

\begin{equation}
    W_\alpha(y, \mathbf{x}) = \sum_{i \in N} \alpha(y - x_i) = \sum_{i \in N} -|y - x_i|^p
\end{equation}

\subsubsection{Special Cases: Utilitarian and Egalitarian Welfare}

Both \textit{utilitarian welfare} and \textit{egalitarian welfare} are special cases of the $p$-mean utility functions, fitting naturally within our framework. The utilitarian welfare function corresponds to the case where \( p = 1 \). Similarly, The egalitarian welfare function corresponds to the limiting case as \( p \rightarrow \infty \).

\subsection{Nash Welfare}

In this section, we demonstrate that the \textit{Nash welfare function} ~\cite{lam2023nash} is also fully compatible with our framework. 
To formalize this, let \( y_{Nash} \) denote the facility location that maximizes the Nash welfare. We express this as:

\begin{equation}
    y_{Nash} = \arg \max_{y \in [0,1]} \prod_{i \in N} (1 - |y - x_i|)
\end{equation}

 We can simplify this by applying the logarithmic transformation. Since the logarithmic function is monotonic, it preserves the location of the maximum. Therefore, we have:

\begin{equation}
    y_{Nash} = \arg \max_{y \in [0,1]} \log \left( \prod_{i \in N} (1 - |y - x_i|) \right)
\end{equation}

Or equivalently:
\begin{equation}
    y_{Nash} = \arg \max_{y \in [0,1]} \sum_{i \in N} \log(1 - |y - x_i|)
\end{equation}

 Thus, by defining the utility function as \( \alpha(x) = \log(1 - |x|) \), the Nash welfare problem is equivalent to maximizing the total welfare function \( W_\alpha(y, \mathbf{x}) \). 

It is worth noting that our framework naturally accommodates asymmetric welfare functions, an area that has been explored in a limited number of facility location studies ~\cite{Leitner2014,Drezner2001}. While the study of asymmetry in FLPs remains relatively sparse, our framework offers a convenient approach to incorporating such welfare functions.

Having established that the \( p \)-mean functions—along with the utilitarian and egalitarian welfare functions—and the Nash welfare function can be effectively incorporated into our \emph{FLIGHT} framework, we now turn our attention to studying the properties of this generalized welfare formulation. In the following section, we examine key structural properties of the generalized welfare function, with particular emphasis on concavity, which is a natural assumption in many practical contexts since utility typically decreases with increasing distance.

\noindent
We choose to use \emph{utility} over \emph{cost} in our exposition, primarily for conceptual clarity. Note that our framework, \ourframework, is inherently flexible and capable of unifying both cost and utility perspectives under the broader notion of \emph{agent single-peaked preferences}. Within this formulation, the $\alpha$-welfare function can be interpreted as an aggregation of these preferences, ensuring a cohesive and generalized approach to welfare optimization.

We proceed by proving several theorems related to these properties, thereby further elucidating the theoretical foundation of the FLIGHT framework.


\
\section{$\alpha$-Welfare: Properties and Computation}
\label{sec:alpha_welfare}

In the previous section, we demonstrated that a wide variety of existing welfare notions can be incorporated into our framework. Here, we delve into the general structural properties of $\alpha$-Welfare functions, focusing on how these properties relate to the computation and approximation of various welfare functions within the framework. Notably, many of these properties echo results found in the literature, thus highlighting the unifying power of our framework. Moreover, several proofs become simplified when viewed through the lens of the generalized framework. We have provided proof sketches wherever possible. The complete proofs are present in {Appendix~\ref{sec:proof_sec}} of this paper.


\subsection{Assumptions}
We begin by assuming that the utility function \( \alpha(x) \) is concave with respect to \( x \), which captures the phenomenon of diminishing returns as the distance between the facility and an agent increases. This assumption of concavity serves as the foundation for the theorems presented in the subsequent sections. Additionally, we assume that \( \alpha(x) \) attains its maximum at \( x = 0 \), reflecting the highest utility when the agent is located at the facility.

The following properties arise naturally from a fundamental assumption of concave utility functions. Note that these do not need to be imposed as design choices.

\subsection{Properties of $\alpha$-Welfare}

\begin{theorem}
\label{thm:concave}
The total welfare function $W_\alpha(y, \mathbf{x})$, is concave in \( y \).
\end{theorem}

\textit{Proof}. The total welfare is defined as the sum of individual utility functions: 
\[
W_\alpha(y, \mathbf{x}) = \sum_{i \in N} \alpha(y - x_i)
\]
Since $\alpha(x)$ is concave, and the sum of {finitely many} concave functions is also concave, it follows that $W_\alpha(y, \mathbf{x})$ is concave in \( y \). 
\qed\\

The significance of Theorem \ref{thm:concave} lies in the fact that the concavity of the total welfare function \( W_\alpha(y, \mathbf{x}) \) implies it is \textit{single-peaked} with respect to \( y \), a property that is analogous to the behavior observed in Nash welfare functions~\cite{lam2023nash}. Furthermore, the structural properties established in Theorems \ref{thm:locinv}, \ref{thm:shift}, and \ref{thm:max_shift} exhibit similar characteristics to those studied in~\cite{lam2023nash}, reinforcing the parallels between our framework and Nash welfare-based approaches in facility location.
. Single-peakedness is crucial in optimization. It allows the use of efficient convex or concave optimization algorithms to locate the maximum welfare point, facilitating computational approaches to solving the facility location problem. Furthermore, the concavity guarantees that local maxima are also global maxima, simplifying the analysis and solution of the problem. {Our next Theorem \ref{thm:locinv} proves the location invariance property of \ourframework.}

\begin{theorem}
\label{thm:locinv}
Let \( \mathbf{x} = (x_1, x_2, \dots, x_n) \) and \( \mathbf{x'} = (x_1 + c, x_2 + c, \dots, x_n + c) \) be two location profiles, where \( c \in \mathbb{R} \) represents a constant shift. Then, the following holds:
\[
W_\alpha(y, \mathbf{x'}) = W_\alpha(y - c, \mathbf{x}),
\]
and consequently, 
\[
P_\alpha(\mathbf{x'}) = P_\alpha(\mathbf{x}) + c.
\]
\end{theorem}

\textit{Proof.} The key idea behind this result is that shifting all the agents locations by a constant \( c \) results in a corresponding shift in the facility location by \( c \), without affecting the total welfare function. 

To see this, consider the welfare function \( W_\alpha(y, \mathbf{x}) = \sum_{i \in N} \alpha(y - x_i) \), where \( \alpha(z) \) is the utility function. When the agent profile \( \mathbf{x} \) is shifted by a constant \( c \), the welfare function for the shifted profile becomes:
\[
W_\alpha(y, \mathbf{x'}) = \sum_{i \in N} \alpha((y - c) - x_i).
\]
This is equivalent to the original welfare function with the facility location adjusted by \( c \), i.e., \( W_\alpha(y, \mathbf{x'}) = W_\alpha(y - c, \mathbf{x}) \).

Thus, the optimal facility location for the shifted profile, \( P_\alpha(\mathbf{x'}) \), is simply the original location shifted by \( c \), i.e., \( P_\alpha(\mathbf{x'}) = P_\alpha(\mathbf{x}) + c \). This completes the proof.
\qed


Theorem \ref{thm:shift} shows how the movement of a single agent $x_i$ by a constant affects the peak $P_\alpha$ {with utility function $\alpha$}.

\begin{theorem}
\label{thm:shift}
Let $\mathbf{x} = (x_1, \dots, x_n)$ be the agent location profile. If agent $x_i$ is shifted left by a constant $c \in (0, x_i]$, resulting in a new profile $\mathbf{x}' = (x_1, \dots, x_i - c, \dots, x_n)$, then:
\[
P_\alpha(\mathbf{x}') \leq P_\alpha(\mathbf{x})
\]
\end{theorem}


\textit{Proof}. Since \( W_\alpha(y, \mathbf{x'}) \) is concave and single-peaked, to prove that \( P_\alpha(\mathbf{x}') \leq P_\alpha(\mathbf{x}) \), it suffices to show that \( W_\alpha(y, \mathbf{x'}) \) 
is strictly decreasing in the interval \( (P_\alpha(\mathbf{x}), 1) \). This is because \( P_\alpha(\mathbf{x}) \) is the maximum of \( W_\alpha(y, \mathbf{x}) \), and proving the function decreases beyond this point implies that the new maximum \( P_\alpha(\mathbf{x}') \) must occur at a lower value.

We begin by expressing the total welfare for the perturbed location profile \( \mathbf{x}' \):

\[
W_\alpha(y, \mathbf{x}') = \alpha(y - x_1) + \dots + \alpha(y - (x_i - c)) + \dots + \alpha(y - x_n)
\]

\noindent This can be rewritten as:

\[
W_\alpha(y, \mathbf{x}') = W_\alpha(y, \mathbf{x}) + \alpha(y - (x_i - c)) - \alpha(y - x_i)
\]

\noindent Next, we differentiate this expression with respect to \( y \):

\[
\frac{dW_\alpha(y, \mathbf{x}')}{dy} = \frac{dW_\alpha(y, \mathbf{x})}{dy} + \frac{d\alpha}{dy}\bigg|_{y - (x_i - c)} - \frac{d\alpha}{dy}\bigg|_{y - x_i}
\]

\noindent Now, consider the terms on the right-hand side:

1. Since \( W_\alpha(y, \mathbf{x}) \) is concave and \( P_\alpha(\mathbf{x}) \) is its maximum, we know that \( \frac{dW_\alpha(y, \mathbf{x})}{dy} < 0 \) for \( y \in (P_\alpha(\mathbf{x}), 1) \).
   
2. Additionally, because \( \alpha(x) \) is concave, we have \( \frac{d\alpha}{dy} \big|_{y - (x_i - c)} \leq \frac{d\alpha}{dy} \big|_{y - x_i} \). This follows from the fact that the derivative of a concave function decreases as the input increases.

3. As a result, \( \frac{d\alpha}{dy} \big|_{y - (x_i - c)} - \frac{d\alpha}{dy} \big|_{y - x_i} \leq 0 \).

Thus, the total derivative \( \frac{dW_\alpha(y, \mathbf{x'})}{dy} \) remains negative in the interval \( (P_\alpha(\mathbf{x}), 1) \). Since the function is decreasing beyond \( P_\alpha(\mathbf{x}) \), it follows that:
$P_\alpha(\mathbf{x}') \leq P_\alpha(\mathbf{x})$
\qed

It is important to observe that when an agent moves to the right, i.e., \( x_i \rightarrow x_i + c \) with \( c \geq 0 \), we can prove by symmetry arguments that \( P_\alpha(\mathbf{x}) \leq P_\alpha(\mathbf{x'}) \), where \( \mathbf{x'} \) represents the updated location profile after the shift. 

This result follows by considering \( \mathbf{x'} \) as the initial location profile and \( \mathbf{x} \) as the deviated profile. Then, by applying Theorem 3, which states that the optimal facility location shifts in the direction of the agent's movement, we conclude that \( P_\alpha(\mathbf{x}) \leq P_\alpha(\mathbf{x'}) \). 

The importance of this theorem lies in the fact that it establishes a directional monotonicity property. Specifically, assuming the positions of all other agents remain fixed, if a single agent shifts to the right, the peak \( P_{\alpha} \) will not move to the left, and similarly, if the agent shifts to the left, the peak will not move to the right. 

We now address a broader question: how does the optimal facility location \( P_\alpha \) change when all agents move from one location profile \( \mathbf{x} \) to a new profile \( \mathbf{x}' \)? Specifically, we aim to understand the relationship between the shifts in individual agent positions and the resulting change in the welfare-maximizing facility location.

\begin{theorem}
\label{thm:max_shift}
For any two agent location profiles $\mathbf{x} = (x_1, x_2 \dots, x_n)$ and $\mathbf{x}' = (x_1', x_2', \dots, x_n')$, the following inequality holds:
\[
|P_\alpha(\mathbf{x}) - P_\alpha(\mathbf{x}')| \leq \max_{i \in [n]} |x_i - x_i'|
\]
\end{theorem}


\textit{Proof}.\footnote{Note that (except for the \ourframework\ setting) this proof is almost exactly equivalent to the one in \cite{lam2023nash}. }
Define two new location profiles, \( \mathbf{x}_{+c} \) and \( \mathbf{x}_{-c} \), as follows:
\[
\mathbf{x}_{+c} \triangleq (x_1 + c, x_2 + c, \dots, x_n + c)
\]
\[
\mathbf{x}_{-c} \triangleq (x_1 - c, x_2 - c, \dots, x_n - c)
\]

By Theorem~\ref{thm:shift}, we know:
\[
P_\alpha(\mathbf{x}_{-c}) \leq P_\alpha(\mathbf{x}') \leq P_\alpha(\mathbf{x}_{+c})
\]

Additionally, by Theorem~\ref{thm:locinv}, we have:
\[
P_\alpha(\mathbf{x}) - c \leq P_\alpha(\mathbf{x}') \leq P_\alpha(\mathbf{x}) + c
\]

This implies:
\[
-c \leq P_\alpha(\mathbf{x}') - P_\alpha(\mathbf{x}) \leq +c
\]

Therefore, we conclude:
\[
|P_\alpha(\mathbf{x}') - P_\alpha(\mathbf{x})| \leq c = \max_{i \in N} |x_i' - x_i|
\]

\qed
\subsection{Approximation Bounds: Comparison to Other Welfare Metrics}

We now turn to the question of how well different welfare functions approximate one another under this framework. 


\begin{theorem}
Let \( \alpha \) be a utility function with an upper Lipschitz constant \( \lambda_\alpha^u \) and a lower Lipschitz constant \( \lambda_\alpha^d \). Additionally, assume that \( \alpha(x) > 0 \) for all \( x \in [-1, 1] \), even though agent locations are restricted to $[0,1]$. Define:
\[
D_\alpha = \min\left\{ \alpha(0) + (n-1)\alpha(-1), \alpha(0) + (n-1)\alpha(1) \right\}.
\]
Then, for all \( y \in [0, 1] \), the following inequality holds:
\[
e^{\frac{\lambda_\alpha^d (P_\alpha(\mathbf{x}) - y)}{ \alpha(0)}} 
\leq 
\frac{W_\alpha(P_\alpha(\mathbf{x}), \mathbf{x})}{W_\alpha(y, \mathbf{x})} 
\leq 
e^{\frac{n \lambda_\alpha^u (P_\alpha(\mathbf{x}) - y)}{D_\alpha}}.
\]
\end{theorem}


\textit{Proof sketch.} To establish the desired bounds, we leverage the Lipschitz properties of the logarithmic function \( \ln(x) \), as well as the Lipschitz continuity of \( \alpha(x) \), to derive a bound on \( \ln\left(\frac{W_\alpha(P_\alpha(\mathbf{x}), \mathbf{x})}{W_\alpha(y, \mathbf{x})}\right) \).

The total welfare function \( W_\alpha(y, \mathbf{x}) \) is the sum of individual utilities based on the location of the facility at \( y \). By applying the Lipschitz properties of \( \alpha(x) \), we can bound the ratio of welfare functions by exponentiating the bound on their logarithmic difference. Formally, we write:
\[
\ln\left(\frac{W_\alpha(P_\alpha(\mathbf{x}), \mathbf{x})}{W_\alpha(y, \mathbf{x})}\right) = \ln(W_\alpha(P_\alpha(\mathbf{x}), \mathbf{x})) - \ln(W_\alpha(y, \mathbf{x})).
\]
Using the Lipschitz property of \( \ln(x) \) and the fact that \( \alpha(x) \) is concave with its maximum at \( x = 0 \) 
, we can apply bounds on this difference. Specifically, since \( \alpha(x) \) is Lipschitz continuous with upper and lower bounds given by \( \lambda_\alpha^u \) and \( \lambda_\alpha^d \), we obtain the bounds for the ratio of the welfare functions.

Finally, applying the exponent to both sides of the inequality provides the result, where the upper and lower bounds depend on the constants \( \lambda_\alpha^u \), \( \lambda_\alpha^d \), and the behavior of \( \alpha(x) \) at the extremes of its domain.
\qed\\

Given that the exponent is a rational function with polynomials of the same degree in both the numerator and the denominator, we can conclude (as we prove in Lemma 6) that as \( n \to \infty \), the exponent is asymptotically bounded by a constant. Furthermore, we can assert that the ratio remains sub-exponential even for small values of \( n \).


\begin{lemma}
\label{lem:limit}
As the number of agents $n$ increases, the upper bound on the approximation ratio converges to a constant i.e. at the most $e^{\frac{\lambda_\alpha^u}{\min\{\alpha(-1), \alpha(1)\}}}$.
\end{lemma}
\textit{Proof.}
From the previous theorem, the upper bound on the approximation ratio is given by:
\[
\frac{W_\alpha(P_\alpha(\mathbf{x}), \mathbf{x})}{W_\alpha(y, \mathbf{x})} \leq e^{\frac{n \lambda_\alpha^u (P_\alpha(\mathbf{x}) - y)}{D_\alpha}}.
\]
As \( n \to \infty \), the denominator behaves as\footnote{Note that $\alpha(x) > 0~\forall x \in [-1,1] $.}:
\[
D_\alpha \sim (n-1) \min\{\alpha(-1), \alpha(1)\}.
\]
Thus, the approximation ratio becomes:
\[
\lim_{n \to \infty} e^{\frac{n \lambda_\alpha^u (P_\alpha(\mathbf{x}) - y)}{(n-1) \min\{\alpha(-1), \alpha(1)\}}} = e^{\frac{\lambda_\alpha^u (P_\alpha(\mathbf{x}) - y)}{\min\{\alpha(-1), \alpha(1)\}}}.
\]
However, since $P_\alpha(\mathbf{x}) - y \le 1$, we have:
\[
\lim_{n \to \infty} e^{\frac{n \lambda_\alpha^u (P_\alpha(\mathbf{x}) - y)}{(n-1) \min\{\alpha(-1), \alpha(1)\}}} \le e^{\frac{\lambda_\alpha^u}{\min\{\alpha(-1), \alpha(1)\}}}.
\]

\begin{lemma}
Let $\alpha , \beta$ be two utility functions with corresponding welfares $W_\beta(y, \mathbf{x}) , W_\alpha(y, \mathbf{x})$ and maximizers $P_\beta(\mathbf{x}) , P_\alpha(\mathbf{x})$. Then, \textbf{Theorem 5} yields the approximation ratio between utility functions $\alpha$ and $\beta$:
\[
\frac{W_\alpha(P_\alpha(\mathbf{x}), \mathbf{x})}{W_\alpha(P_\beta(\mathbf{x}), \mathbf{x})} \leq e^{\frac{n \lambda_\alpha (P_\alpha(\mathbf{x}) - P_\beta(\mathbf{x}))}{ D_\alpha}} \leq e^{\frac{n \lambda_\alpha }{ D_\alpha}}
\]
\noindent By Lemma~\ref{lem:limit}, as \( n \to \infty \), this approximation ratio becomes constant. 
\end{lemma}

\textit{Proof:} The proof follows directly from Lemma ~\ref{lem:limit}. The final part of the inequality is true because $P_\alpha(\mathbf{x}) - P_\beta(\mathbf{x}) \le 1$.
\qed

\subsection{Bounding the Distance Between Peaks}

The next theorem provides bound on the distance between the peaks $P_\alpha$ and $median$, based on the configuration of agent locations.

\begin{theorem}
Let \( med \) denote the median of the location profile \( \mathbf{x} = (x_1, x_2, \dots, x_n) \). If \( n \) is even, define the median as:
\[
med = \frac{x_{\frac{n}{2}} + x_{\frac{n}{2}+1}}{2}.
\]
Then, the following inequality holds:
\[
|med - P_\alpha(\mathbf{x})| \leq \frac{1}{2} \max_{i=1}^{\lfloor n/2 \rfloor} |d_i^+ - d_i^-|,
\]
where \( d_i^+ = |med - x_{\lfloor n/2 \rfloor + i}| \) and \( d_i^- = |med - x_{\lceil n/2 \rceil - i}| \), and \( \alpha \) is a symmetric utility function.
\end{theorem}

\textit{Proof sketch}. The central idea behind this proof is recognizing that, if the location profile \( \mathbf{x} \) were perfectly symmetric around the median \( med \), then the peak of the welfare function, \( P_\alpha(\mathbf{x}) \), would coincide with the median due to the symmetry of \( \alpha \).

Thus, the deviation of \( P_\alpha(\mathbf{x}) \) from the median arises solely due to the asymmetry in the distribution of agents around the median. The key question becomes: how far must the agents be shifted to transform the location profile \( \mathbf{x} \) into a symmetric configuration?

To quantify this, we compare the distances of the agents on either side of the median. For each agent positioned at \( x_{\lfloor n/2 \rfloor + i} \) (to the right of the median), we consider the distance \( d_i^+ \) from the median, and similarly, for each agent at \( x_{\lceil n/2 \rceil - i} \) (to the left of the median), we consider the distance \( d_i^- \). 

The difference \( |d_i^+ - d_i^-| \) measures how far these corresponding agents deviate from a symmetric configuration. The maximum of these deviations for all pairs of agents gives a measure of the total asymmetry in the distribution of the agents around the median.

Since the location of \( P_\alpha(\mathbf{x}) \) is influenced by the overall symmetry of the location profile, the deviation of \( P_\alpha(\mathbf{x}) \) from the median is bounded by the total asymmetry, which is expressed as:
\[
|med - P_\alpha(\mathbf{x})| \leq\frac{1}{2} \max_{i=1}^{\lfloor n/2 \rfloor} |d_i^+ - d_i^-|.
\]
This completes the sketch of the proof.
\qed

\begin{figure}[!ht]
    \centering
    \includegraphics[width = \linewidth]{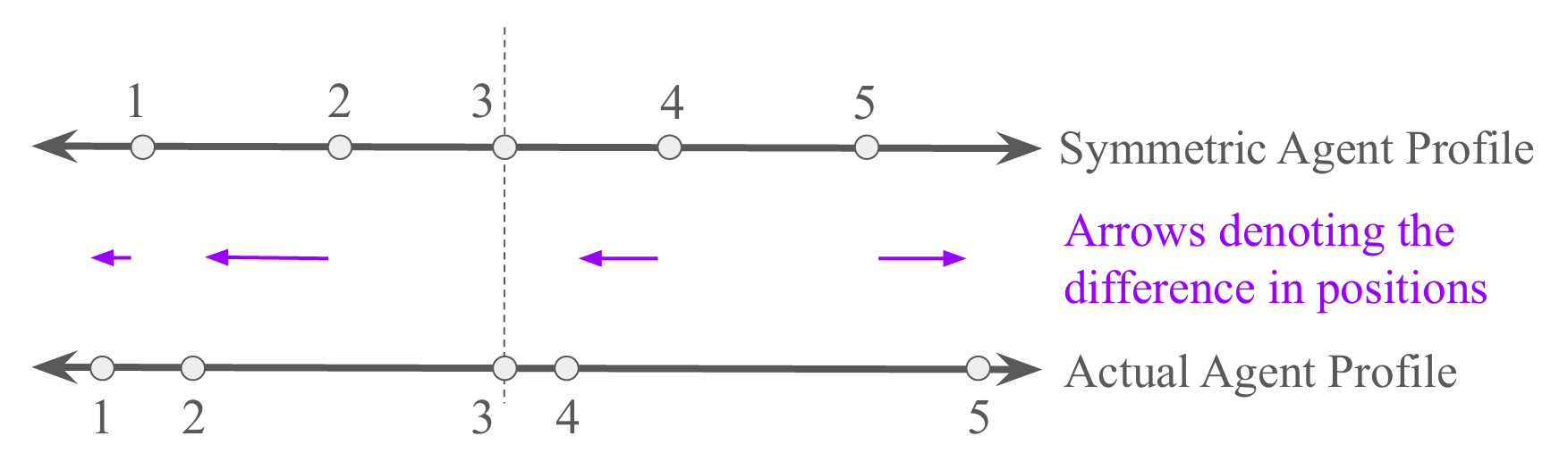}
    \caption{Illustration of Theorem 8: Bounding the Deviation of the Welfare Peak from the Median. The \textcolor{LightViolet}{longest arrow}, representing the maximum deviation from the symmetric agent profile, provides the required upper bound.}

    \label{fig:enter-label}
\end{figure}

In the following section, we will conduct a probabilistic analysis of our {FLIGHT} framework. Specifically, we will consider scenarios in which the agent locations are drawn from a probability distribution. By incorporating this probabilistic perspective, we aim to provide a more comprehensive understanding of how the distribution of agents influences the welfare outcomes within our framework.

\section{Probabilistic Analysis of $\alpha$-Welfare}
\label{sec:prob_alpha_welfare}

In this section, we extend our analysis by introducing a probabilistic framework in which agent locations are treated as random variables. Specifically, we assume that the agents' preferred locations \(x_i\) are i.i.d. samples from a probability distribution \(\mathcal{P}\), i.e., \(x_i \sim \mathcal{P}\). This formulation allows us to examine the behavior of welfare functions when agent positions are drawn from a probabilistic distribution, which is particularly useful in real-world scenarios where exact agent locations may be uncertain.

\subsection{Expected Welfare Function}
We define\footnote{In scenarios where agent locations are probabilistic, the total welfare function becomes a random variable. Accordingly, we extend our prior definitions in Section~\ref{sec:flight_frmwrk} to accommodate this stochastic setting. Such analysis provides insights into ex-ante performance.} the expected total welfare, \( \mathbb{W}_\alpha^\mathcal{P}(y, \mathbf{x}) \), as the expected welfare at location \( y \) when agents are sampled from the probability distribution \( \mathcal{P} \). This expected welfare provides a probabilistic generalization of our previously defined deterministic welfare functions.

\begin{definition}[Empirical Welfare]
The empirical welfare function \( W_\alpha(y, \mathbf{x}) \) is defined as the sum of individual utilities for agents located at \( \mathbf{x} = (x_1, x_2, \dots, x_n) \), where the locations \( x_i \) are sampled from the probability distribution \( \mathcal{P} \). Formally:
\[
W_\alpha(y, \mathbf{x}) = \sum_{i=1}^{n} \alpha(y - x_i).
\]
\end{definition}

\begin{definition}[Expected Welfare]
Let \( \mathbb{W}_\alpha^\mathcal{P}(y, \mathbf{x}) \) represent the expected welfare for agents sampled i.i.d. from the distribution \( \mathcal{P} \). Formally, this is given by:
\[
\mathbb{W}_\alpha^\mathcal{P}(y, \mathbf{x}) = \mathbb{E}_{\mathbf{x} \sim \mathcal{P}^n} \left[ W_\alpha(y, \mathbf{x}) \right],
\]
where \( \alpha \) is the individual utility function, \( \mathbf{x} = (x_1, x_2, \dots, x_n) \) represents the agent locations, and \( n \) is the number of agents.
\end{definition}
\begin{theorem}
The expected welfare function \( \mathbb{W}_\alpha^\mathcal{P}(y, \mathbf{x}) \) is given by:
\[
\mathbb{W}_\alpha^\mathcal{P}(y, \mathbf{x}) = n \times [\alpha \circledast \mathcal{P}](y),
\]
where \([\alpha \circledast \mathcal{P}](y)\) represents the convolution of the utility function \( \alpha \) and the probability distribution \( \mathcal{P} \), evaluated at location \( y \).
\end{theorem}

\textit{Proof sketch}.
The result follows from the linearity of expectation and the fact that the sum of \( n \) independent random variables sampled from \( \mathcal{P} \) is equivalent to scaling the expected utility by \( n \). Specifically, since each agent's utility is independently and identically distributed (i.i.d.) from \( \mathcal{P} \), the expected welfare for \( n \) agents can be written as the expected welfare for one agent multiplied by \( n \). This yields the convolution result.
\qed

In many practical applications, especially in large-scale systems, we are often not given the exact locations of all agents but rather a probability distribution \( \mathcal{P} \) that describes their likely positions. When dealing with a large number of agents, directly computing the total welfare based on each individual’s location can become computationally expensive. This challenge necessitates the development of more efficient techniques for estimating the welfare, particularly when a probability distribution is available.

The following two theorems provide insights, allowing us to circumvent the need to compute the welfare function through explicit agent positions. Instead, these results show that we can rely on the probability distribution \( \mathcal{P} \) to derive robust estimates of the welfare function.

First (Theorem 10), we demonstrate that given only the probability distribution \( \mathcal{P} \), the best estimator of the welfare function is the expected welfare. This means that the expected welfare serves as an unbiased estimator with the minimum variance, ensuring the most accurate estimate of the empirical welfare function. This result is highly valuable because it simplifies the computation by focusing on the mean of the distribution rather than requiring the sum over all agents.

Second (Theorem 11), we establish that as the number of agents increases, the empirical welfare converges in probability to the expected welfare. This result, rooted in the weak law of large numbers, guarantees that as the number of agents becomes arbitrarily large, the difference between the empirical and expected welfare diminishes. Hence, for engineering applications where we often deal with large populations, using the expected welfare is not only computationally simpler but also practically equivalent to calculating the empirical welfare.

Together, these results underscore the utility of relying on the expected welfare in large-scale systems. The expected welfare is easier to compute and analyze, and as the number of agents increases, it becomes a reliable proxy for the actual welfare function, making it highly suitable for real-world engineering applications.

\subsection{Optimality in Function Space}

We next examine how the expected welfare function relates to minimizing the distance between welfare functions in a suitable function space.

\begin{definition}
Define the \(\mathcal{F}\)-distance between two functions \(f_1\) and \(f_2\) as:
\[
||f_1 - f_2||_{\mathcal{F}} \triangleq \int_\mathbb{R} \left(f_1(x) - f_2(x)\right)^2 dx
\]
\end{definition}

\begin{theorem}
The function \(\mathbb{W}_\alpha^\mathcal{P}(y, \mathbf{x})\) minimizes the expected \(\mathcal{F}\)-distance to the empirical welfare function \(W_\alpha(y, \mathbf{x})\), i.e., it is the best approximation of the empirical welfare 
in the \(\mathcal{F}\)-distance sense.
\end{theorem}

\textit{Proof sketch}. By definition of expected welfare, \(W_\alpha(y, \mathcal{P}, n)\) is the mean of the empirical welfare function \(W_\alpha(y, \mathbf{x})\), where \(\mathbf{x}\) is the random vector of agent locations. Given that the \(\mathcal{F}\)-distance is a measure of difference between two functions, the expected value of the empirical welfare is the function that minimizes this difference. The proof mirrors the discrete case. (The Minimum Variance Unbiased Estimator of a random variable is simply its mean.)  \qed

\subsection{Asymptotic result on Welfares}

The weak law of large numbers is a classical result in probability theory. Here, we derive a variation of the law applied to our welfare framework.

\begin{theorem}[Asymptotic result for Welfare Functions]
Let \(W_\alpha(y, \mathbf{x})\) be the empirical welfare function for a sample of \(n\) agents and $X \sim \mathcal{P}$. Then, as \(n \to \infty\), the empirical welfare converges in probability to the expected welfare for a single agent:
\[
\frac{W_\alpha(y, \mathbf{x})}{n} \xrightarrow{p} \mathbb{W}_\alpha^\mathcal{P}(y, X)
\]

\end{theorem}

\textit{Proof sketch}. This is a direct consequence of the weak law of large numbers. As the number of agents increases, the average empirical welfare converges to the expected welfare. The factor of \(n\) normalizes the total welfare, ensuring convergence to the expected value for one agent sampled from \(\mathcal{P}\). \qed

\section{Conclusion}
In this paper, we introduced a flexible and unified framework for facility location problems using $\alpha$-welfare functions, demonstrating that various well-known welfare models, such as \textit{utilitarian}, \textit{egalitarian}, and \textit{Nash welfare}, are special cases within this framework. We established key structural properties of these functions, including \textit{concavity}, \textit{location invariance}, and \textit{monotonicity}, which simplify optimization for facility placement. Additionally, by incorporating a probabilistic perspective, where agent locations are modeled as independent samples from a distribution, we derived the expected welfare function and analyzed how welfare functions approximate each other under uncertainty. Furthermore, we provided \textit{approximation bounds} between different welfare functions and introduced the $\mathcal{F}$-distance as a metric for evaluating discrepancies between empirical and expected welfare functions. This comprehensive framework supports robust decision-making in facility location problems and offers significant potential for future research.

\section{Further Work}

This paper has introduced a unifying framework for various welfare functions in facility location problems. However, there remains significant potential for further research, particularly in extending the scope of this framework and deepening the theoretical understanding of its applicability. In this section, we outline several promising directions for future investigation.

\subsection{Sufficiency: Empirical and Theoretical Perspectives}

While we have demonstrated that utilitarian, egalitarian, and Nash welfare functions can be incorporated within our framework, there are many other welfare notions that have not been explored in this paper. For example, concepts such as \textit{approval radius} ~\cite{approvalradius2020} and \textit{obnoxious facility location} ~\cite{Erkut1989} could be unified under our approach. These topics were excluded for simplicity but present an opportunity for further empirical investigation into the sufficiency of the framework in capturing diverse welfare functions.

From a theoretical standpoint, the key question is whether our framework is capable of capturing \emph{all} welfare functions under certain assumptions. This question leads to the exploration of whether our model is comprehensive enough to encompass all possible welfare measures or whether there are exceptions that lie outside the scope of the framework.

One promising avenue for addressing this question is the \textit{Kolmogorov-Arnold representation theorem}, which provides a general representation for multivariate functions ~\cite{Kolmogorov1957}. 
The Kolmogorov-Arnold theorem suggests that functions can be represented as:
\[
f(x_1, x_2, \dots, x_n) = \sum_{i=1}^{2n+1} \varphi_i \left( \sum_{j=1}^{n} \psi_{ij}(x_j) \right)
\]
where \(\varphi_i\) and \(\psi_{ij}\) are continuous univariate functions.
If we interpret the facility location function \( f(x_1, x_2, \dots, x_n) \) as the location of the facility, then by assuming symmetry (i.e., the location function remains invariant under permutations of agent positions), we can potentially simplify the representation of welfare functions. For instance, the following is true ~\cite{zaheer2017deep} for symmetric $f$:
\[
f(x_1, x_2, \dots, x_n) =  \varphi \left( \sum_{j=1}^{n} \psi(x_j) \right)
\]

 This structure bears a resemblance to our framework, where we maximize the sum of utilities:
\[
P_\alpha(\mathbf{x}) = \arg \max_{y \in [0, 1]} \left(\sum_{i \in N} \alpha(y - x_i)\right)
\]
Thus, a theoretical extension would involve proving the equivalence of our framework with this symmetric function representation or identifying welfare functions that cannot be expressed within this model.

\subsection{Strategyproofness Concerns}

Another critical direction for future research involves the \textit{strategyproofness} of the proposed framework. Strategyproofness is an essential property in facility location problems, ensuring that no agent can benefit from misreporting their true location. A rigorous analysis of the strategyproofness properties within our framework, especially in relation to various welfare functions, has yet to be conducted. Investigating whether the framework ensures strategyproof mechanisms or if certain welfare functions lead to strategic manipulation is a key open question.

\subsection{Approximation and Bound Improvements}

This paper also provides approximation bounds for various welfare functions, but there is considerable scope for improving these bounds. Stronger and more precise bounds can be derived, particularly when the assumption of strictly positive utility is relaxed. Furthermore, approximation bounds can be extended to \textit{probabilistic settings}, where agent locations are treated as random variables drawn from distributions ~\cite{Drezner2001}. Such work could contribute significantly to the practical applicability of the framework in real-world scenarios, where uncertainties in agent positions are common.

Additionally, both deterministic and probabilistic approximations could be enhanced, especially when considering more complex welfare functions or when the utility functions involve diminishing or negative returns.





\bibliography{main}
\bibliographystyle{unsrt}
\balance

\appendix

\section{Notation Table}

\begin{table}[h!]
\centering
\renewcommand{\arraystretch}{1.3} 
\begin{tabularx}{\textwidth}{>{\raggedright\arraybackslash}X >{\raggedright\arraybackslash}X}
    \toprule
    \textbf{Notation} & \textbf{Description} \\
    \midrule
    $N = \{1, \dots, n\}$     & Set of agents \\
    $x_i \in [0, 1]$           & Location of agent $i$ on a line segment \\
    $\mathbf{x} = (x_1, \dots, x_n)$ & Agent location profile \\
    $y$                       & Variable denoting points on [0,1] \\
    $W_{\alpha}(y, \mathbf{x})$& Total welfare function under utility function $\alpha$ \\
    $P_{\alpha}(\mathbf{x})$   & Optimal facility location that maximizes the welfare \\
    $\alpha(x)$                & Utility function reflecting welfare decay over distance \\
    $\mathcal{P}$  & Probability distribution of agents' preferred location\\
    $\mathbb{W}_{\alpha}^\mathcal{P}(y , \mathbf{x} )$& Expected welfare function under utility $\alpha$ and distribution $\mathcal{P}$.\\
    $W_{utilitarian}(y, \mathbf{x})$       & Utilitarian social welfare: $W_{utilitarian}(y, \mathbf{x}) = \sum_{i} (1 - |y- x_i|)$ \\
    $W_{egalitarian}(y, \mathbf{x})$       & Egalitarian social welfare: $W_{egalitarian}(y, \mathbf{x}) = \min_{i} (1 - |y- x_i|)$ \\
    $W_{Nash}(y, \mathbf{x})$      & Nash social welfare: $W_{Nash}(y, \mathbf{x}) = \prod_{i} (1 - |y- x_i|)$ \\
    \bottomrule
\end{tabularx}
\caption{Notations Used in the Facility Location Problem}
\label{tab:notations}
\end{table}

\section{Theorem Proofs}
\label{sec:proof_sec}
\subsection{Theorem 1}
The total welfare function \( W_\alpha(y, \mathbf{x}) \), is concave in \( y \).
\begin{proof}
The total welfare function \( W_\alpha(y, \mathbf{x}) \) is defined as the sum of the individual utilities \( \alpha(y - x_i) \) for each agent \( i \in N \):
\[
W_\alpha(y, \mathbf{x}) = \sum_{i \in N} \alpha(y - x_i),
\]
where \( \alpha(z) \) is a concave function. Since \( W_\alpha(y, \mathbf{x}) \) is the sum of concave functions, we can consider its negative, \( -W_\alpha(y, \mathbf{x}) \):
\[
-W_\alpha(y, \mathbf{x}) = -\sum_{i \in N} \alpha(y - x_i).
\]
The negative of a concave function is convex, and the sum of convex functions is also convex. Therefore, \( -W_\alpha(y, \mathbf{x}) \) is convex. Since the negative of a convex function is concave, it follows that \( W_\alpha(y, \mathbf{x}) \) is concave in \( y \).
\end{proof}

\subsection{Theorem 2}

Let \( \mathbf{x} = (x_1, x_2, \dots, x_n) \) and \( \mathbf{x'} = (x_1 + c, x_2 + c, \dots, x_n + c) \) be two location profiles, where \( c \in \mathbb{R} \) represents a constant shift. Then, the following holds:
\[
W_\alpha(y, \mathbf{x'}) = W_\alpha(y - c, \mathbf{x}),
\]
and consequently, 
\[
P_\alpha(\mathbf{x'}) = P_\alpha(\mathbf{x}) + c.
\]

\begin{proof}
We begin by considering the total welfare function \( W_\alpha(y, \mathbf{x}) \), which is defined as the sum of individual utilities for agents located at \( \mathbf{x} = (x_1, x_2, \dots, x_n) \) with the facility placed at \( y \):
\[
W_\alpha(y, \mathbf{x}) = \sum_{i \in N} \alpha(y - x_i),
\]
where \( \alpha(z) \) denotes the utility function based on the distance \( z = y - x_i \).

Now, consider the shifted location profile \( \mathbf{x'} = (x_1 + c, x_2 + c, \dots, x_n + c) \), where all agents are shifted by a constant \( c \in \mathbb{R} \). The total welfare function for this shifted profile is given by:
\[
W_\alpha(y, \mathbf{x'}) = \sum_{i \in N} \alpha(y - (x_i + c)) = \sum_{i \in N} \alpha((y - c) - x_i).
\]
This expression is equivalent to the total welfare function for the original profile \( \mathbf{x} \), but with the facility located at \( y - c \):
\[
W_\alpha(y, \mathbf{x'}) = W_\alpha(y - c, \mathbf{x}).
\]
Thus, the total welfare function for the shifted profile \( \mathbf{x'} \) is equivalent to the total welfare function for the original profile \( \mathbf{x} \), with the facility location adjusted by \( c \).

Next, we consider the optimal facility location \( P_\alpha(\mathbf{x}) \), which is the value of \( y \) that maximizes the total welfare function \( W_\alpha(y, \mathbf{x}) \):
\[
P_\alpha(\mathbf{x}) = \arg \max_y W_\alpha(y, \mathbf{x}).
\]
By the earlier result, we know that shifting the agent locations by \( c \) corresponds to shifting the facility location by \( c \) as well. Hence, the optimal facility location for the shifted profile \( \mathbf{x'} \) satisfies:
\[
P_\alpha(\mathbf{x'}) = P_\alpha(\mathbf{x}) + c.
\]

In conclusion, we have shown that shifting the agent locations by a constant \( c \) results in a corresponding shift in the optimal facility location by \( c \), and the total welfare function is preserved under this transformation.
\end{proof}

\subsection{Theorem 3}
Proof in the paper. 
\subsection{Theorem 4}
Proof in the paper.
\subsection{Theorem 5}

Let $\alpha$ be a utility function with upper Lipschitz constant $\lambda_\alpha^u$ and lower Lipshitz constant $\lambda_\alpha^d$ Additionally, assume that $\alpha(x) > 0$ for all $x \in [-1,1]$. Then, for all \( y \in [0,1] \), the following inequality holds:
\[
e^{\frac{\lambda_\alpha ^d(P_\alpha(\mathbf{x}) - y)}{n \alpha(0)}} \leq \frac{W_\alpha(P_\alpha(\mathbf{x}), \mathbf{x})}{W_\alpha(y, \mathbf{x})} \leq e^{\frac{n \lambda_\alpha^u (P_\alpha(\mathbf{x}) - y)}{ min\{\alpha(0) + (n-1)\alpha(-1) , \alpha(0) + (n-1)\alpha(1)\}}}
\]

\textit{Proof}. We know that if we constrict the $ln(x)$ function to the domain $(a,b)$, then the upper Lipshitz would be $\frac{1}{a}$ and the lower Lipshitz constant would be $\frac{1}{b}$.

So, we make the following expansion:

\[ln(\frac{W_\alpha(P_\alpha(\mathbf{x}) , \mathbf{x})}{W_\alpha(y , \mathbf{x})})  = \frac{ln(W_\alpha(P_\alpha(\mathbf{x}) , \mathbf{x})) -  ln(W_\alpha(y , \mathbf{x}))}{W_\alpha(P_\alpha(\mathbf{x}) , \mathbf{x}) - W_\alpha(y , \mathbf{x})} \times \frac{W_\alpha(P_\alpha(\mathbf{x}) , \mathbf{x}) - W_\alpha(y , \mathbf{x})}{P_\alpha(\mathbf{x}) - y} \times ({P_\alpha(\mathbf{x}) - y})\]
\subsubsection*{Term 1}
Let us examine the first term of the right hand side of the equation. By the Lipshitz constant of the log, we have:
\[ \frac{1}{max_{y\in[0,1] 
 }W_\alpha(y , \mathbf{x})}\le\frac{ln(W_\alpha(P_\alpha(\mathbf{x}) , \mathbf{x})) -  ln(W_\alpha(y , \mathbf{x}))}{W_\alpha(P_\alpha(\mathbf{x}) , \mathbf{x}) - W_\alpha(y , \mathbf{x})} \le \frac{1}{min_{y\in[0,1] 
 }W_\alpha(y , \mathbf{x})} \]

But, we know:
\[max_{y\in[0,1] 
 }W_\alpha(y , \mathbf{x}) = \alpha(y - x_1) +\hdots + \alpha(y-x_n)  \le n \alpha(0)\]
 (With no loss of generality, assume $x_1 = 0$ and we know that $\alpha(x)$ peaks at $x = 0$.)\\
 And

\[min_{y\in[0,1] 
 }W_\alpha(y , \mathbf{x}) = min\{W_\alpha(0 , \mathbf{x}) , W_\alpha(1 , \mathbf{x})\}\]
 This is because $W_\alpha(y , \mathbf{x})$ is concave, so the minima are at the corners.
 So, we know:
\[W_\alpha(0 , \mathbf{x}) \ge \alpha(0) + (n-1) \alpha(-1)\]
and 
\[W_\alpha(1 , \mathbf{x}) \ge \alpha(0) + (n-1) \alpha(1)\]

So, we can say:

\[\frac{1}{n \alpha(0)}\le \frac{1}{max_{y\in[0,1] 
 }W_\alpha(y , \mathbf{x})}\le\frac{ln(W_\alpha(P_\alpha(\mathbf{x}) , \mathbf{x})) -  ln(W_\alpha(y , \mathbf{x}))}{W_\alpha(P_\alpha(\mathbf{x}) , \mathbf{x}) - W_\alpha(y , \mathbf{x})}
 \]
 \[\le \frac{1}{min_{y\in[0,1] 
 }W_\alpha(y , \mathbf{x})} \le \frac{1}{min\{\alpha(0) + (n-1) \alpha(-1) , \alpha(0) + (n-1) \alpha(1) \}}\]

\subsubsection*{Term 2}
Let us now consider the second term. We know, by definition of Lipshitz property, that:
\[ \lambda_\alpha^d\le \frac{\alpha(P_\alpha(\mathbf{x})  - x_i) - \alpha(y - x_i)}{ P_\alpha(\mathbf{x}) - y} \le \lambda_\alpha^u~ ~~\forall i \in N\]
This implies:
\[ n\lambda_\alpha^d\le\frac{\sum_{i \in N} \alpha(P_\alpha(\mathbf{x})  - x_i) - \sum_{i \in N}\alpha(y - x_i)}{ P_\alpha(\mathbf{x}) - y}\le n\lambda_\alpha^u~~~~\forall i \in N\]
Or equivalently:
\[ n\lambda_\alpha^d\le \frac{W_\alpha(P_\alpha(\mathbf{x})  - x_i) - W_\alpha(y - x_i)}{ P_\alpha(\mathbf{x}) - y} \le n\lambda_\alpha^u~ ~~\forall i \in N\]

\subsubsection*{Final Product}
Now, let us combine all the terms to get:
\[\frac{n\lambda_\alpha^d}{n\alpha(0)}( P_\alpha(\mathbf{x}) - y)\le ln(\frac{W_\alpha(P_\alpha(\mathbf{x}) , \mathbf{x})}{W_\alpha(y , \mathbf{x})}) \le \frac{n\lambda_\alpha^u}{D_\alpha}(P_\alpha(\mathbf{x}) - y)\]

\qed
\subsection{Theorem 8}

Let \( med \) denote the median of the location profile \( \mathbf{x} = (x_1, x_2, \dots, x_n) \). If \( n \) is even, define the median as:
\[
med = \frac{x_{\frac{n}{2}} + x_{\frac{n}{2}+1}}{2}.
\]
Then, the following inequality holds:
\[
|med - P_\alpha(\mathbf{x})| \leq \frac{1}{2} \max_{i=1}^{\lfloor n/2 \rfloor} |d_i^+ - d_i^-|,
\]
where \( d_i^+ = |med - x_{\lfloor n/2 \rfloor + i}| \) and \( d_i^- = |med - x_{\lceil n/2 \rceil - i}| \), and \( \alpha \) is a symmetric utility function.

\begin{proof}
We begin by considering the median \( med \) of the agent location profile \( \mathbf{x} = (x_1, x_2, \dots, x_n) \). If the distribution of the agents were perfectly symmetric around the median, the optimal facility location \( P_\alpha(\mathbf{x}) \), which maximizes the total welfare, would coincide with the median due to the symmetry of the utility function \( \alpha \).

Since \( \alpha \) is symmetric, any deviation of \( P_\alpha(\mathbf{x}) \) from the median arises due to asymmetry in the agent distribution. To analyze this, we define the deviations of the agents from the median on both sides of \( med \). For each agent on the right of the median, positioned at \( x_{\lfloor n/2 \rfloor + i} \), we define the distance:
\[
d_i^+ = |med - x_{\lfloor n/2 \rfloor + i}|.
\]
Similarly, for each agent on the left of the median, positioned at \( x_{\lceil n/2 \rceil - i} \), we define:
\[
d_i^- = |med - x_{\lceil n/2 \rceil - i}|.
\]

The differences \( |d_i^+ - d_i^-| \) measure the extent to which the distribution is asymmetric around the median. To quantify the maximum deviation from symmetry, we define a \textit{reflection distribution}:

Let \( \mathbf{x}_{\text{ref}} \) denote the symmetric distribution obtained by reflecting the agents to form a symmetric profile around \( med \). Specifically, for each agent on the right \( x_{\lfloor n/2 \rfloor + i} \), reflect this agent to the left side of \( med \) to obtain the position \( med - d_i^+ \). Similarly, for each agent on the left \( x_{\lceil n/2 \rceil - i} \), reflect this agent to the right side of \( med \) to obtain the position \( med + d_i^- \).

Next, we construct the \textit{midpoint symmetric distribution} \( \mathbf{x}_{\text{sym}} \), defined as the set of midpoints between each actual point and its corresponding reflection in the distribution. Formally, for each \( i \), the midpoint between an agent at \( x_{\lfloor n/2 \rfloor + i} \) and its reflected counterpart is given by:
\[
x_{\text{sym}}^{\lfloor n/2 \rfloor + i} = x_{\lfloor n/2 \rfloor + i} + \frac{(d_i^- - d_i^+)}{2}.
\]
and for each agent on the left, the corresponding midpoint is:
\[
x_{\text{sym}}^{\lceil n/2 \rceil - i} = x_{\lceil n/2 \rceil - i} - \frac{(d_i^+ - d_i^-)}{2}.
\]
This ensures that \( \mathbf{x}_{\text{sym}} \) represents the closest symmetric profile to \( \mathbf{x} \), formed by averaging each point and its reflection.

The maximum distance any point needs to travel to reach this symmetric distribution is half of the difference between \( d_i^+ \) and \( d_i^- \), i.e.,
\[
\frac{1}{2} |d_i^+ - d_i^-|.
\]

Now, since the facility location \( P_\alpha(\mathbf{x}) \) depends on the symmetry of the location profile, the deviation of \( P_\alpha(\mathbf{x}) \) from the median \( med \) is bounded by this maximum distance of asymmetry:
\[
|med - P_\alpha(\mathbf{x})| \leq \frac{1}{2} \max_{i=1}^{\lfloor n/2 \rfloor} |d_i^+ - d_i^-|.
\]
This completes the proof.
\end{proof}

\subsection{Theorem 9}

The expected welfare function \( W_\alpha(y, \mathcal{P}, n) \) for \( n \) agents, where agent positions are sampled i.i.d. from a probability distribution \( \mathcal{P} \), is given by:
\[
W_\alpha(y, \mathcal{P}, n) = n \times [\alpha \circledast \mathcal{P}](y),
\]
where \([\alpha \circledast \mathcal{P}](y)\) denotes the convolution of the utility function \(\alpha\) and the probability distribution \(\mathcal{P}\), evaluated at point \(y\).

\begin{proof}
Let the agent locations \( x_1, x_2, \dots, x_n \) be i.i.d. random variables drawn from the probability distribution \( \mathcal{P} \), and let the utility function for an agent located at \( x_i \) be given by \( \alpha(y - x_i) \), where \( y \) represents the facility location. The total welfare for \( n \) agents, denoted by \( W_\alpha(y, \mathbf{x}) \), is the sum of the individual utilities:
\[
W_\alpha(y, \mathbf{x}) = \sum_{i=1}^{n} \alpha(y - x_i).
\]
Taking the expectation of \( W_\alpha(y, \mathbf{x}) \) with respect to the agent positions sampled from \( \mathcal{P} \), we obtain the expected welfare function:
\[
W_\alpha(y, \mathcal{P}, n) = \mathbb{E}\left[ \sum_{i=1}^{n} \alpha(y - x_i) \right].
\]
By the linearity of expectation, we can rewrite this as:
\[
W_\alpha(y, \mathcal{P}, n) = \sum_{i=1}^{n} \mathbb{E}[\alpha(y - x_i)].
\]
Since the agents' positions \( x_i \) are i.i.d. random variables drawn from \( \mathcal{P} \), the expected utility for each agent is the same and given by the convolution of \( \alpha \) and \( \mathcal{P} \). Thus, we have:
\[
\mathbb{E}[\alpha(y - x_i)] = \int_{-\infty}^{\infty} \alpha(y - z) \, \mathcal{P}(z) \, dz = [\alpha \circledast \mathcal{P}](y).
\]
Therefore, the expected welfare function simplifies to:
\[
W_\alpha(y, \mathcal{P}, n) = \sum_{i=1}^{n} [\alpha \circledast \mathcal{P}](y) = n \times [\alpha \circledast \mathcal{P}](y).
\]
This completes the proof.
\end{proof}

\subsection{Theorem 10}

\begin{theorem}
The expected welfare function \( \mathbb{W}_\alpha^\mathcal{P}(y, \mathbf{x}) \) minimizes the expected \( \mathcal{F} \)-distance to the empirical welfare function \( W_\alpha(y, \mathbf{x}) \), i.e., it is the best approximation of the empirical welfare in the \( \mathcal{F} \)-distance sense.
\end{theorem}

\begin{proof}
Let the empirical welfare function be denoted by \( W_\alpha(y, \mathbf{x}) = \sum_{i=1}^{n} \alpha(y - x_i) \), where \( x_i \sim \mathcal{P} \) are i.i.d. random variables sampled from the distribution \( \mathcal{P} \), and let the expected welfare function be \( \mathbb{W}_\alpha^\mathcal{P}(y, \mathbf{x}) = n \times [\alpha \circledast \mathcal{P}](y) \).

We aim to prove that the expected welfare function minimizes the expected \( \mathcal{F} \)-distance between the empirical welfare \( W_\alpha(y, \mathbf{x}) \) and any arbitrary function \( f(y) \).

First, define the \( \mathcal{F} \)-distance between the empirical welfare \( W_\alpha(y, \mathbf{x}) \) and an arbitrary function \( f(y) \) as:
\[
\mathcal{F}(f, W_\alpha) = \mathbb{E}_{x_i \sim \mathcal{P}} \left[ (W_\alpha(y, \mathbf{x}) - f(y))^2 \right].
\]
We seek to minimize this expected \( \mathcal{F} \)-distance. Expanding the expression for \( \mathcal{F}(f, W_\alpha) \), we get:
\[
\mathcal{F}(f, W_\alpha) = \mathbb{E}_{x_i \sim \mathcal{P}} \left[ \left( \sum_{i=1}^{n} \alpha(y - x_i) - f(y) \right)^2 \right].
\]
Expanding the square inside the expectation:
\[
\mathcal{F}(f, W_\alpha) = \mathbb{E}_{x_i \sim \mathcal{P}} \left[ \left( \sum_{i=1}^{n} \alpha(y - x_i) \right)^2 - 2 f(y) \sum_{i=1}^{n} \alpha(y - x_i) + f(y)^2 \right].
\]
Using the linearity of expectation, we can rewrite this as:
\[
\mathcal{F}(f, W_\alpha) = \mathbb{E}_{x_i \sim \mathcal{P}} \left[ \left( \sum_{i=1}^{n} \alpha(y - x_i) \right)^2 \right] - 2 f(y) \mathbb{E}_{x_i \sim \mathcal{P}} \left[ \sum_{i=1}^{n} \alpha(y - x_i) \right] + f(y)^2.
\]
Let us focus on minimizing the expected \( \mathcal{F} \)-distance. The term \( \mathbb{E}_{x_i \sim \mathcal{P}} \left[ \sum_{i=1}^{n} \alpha(y - x_i) \right] \) is simply the expected welfare function, i.e., \( \mathbb{W}_\alpha^\mathcal{P}(y, \mathbf{x}) \). Thus, we have:
\[
\mathcal{F}(f, W_\alpha) = \mathbb{E}_{x_i \sim \mathcal{P}} \left[ \left( \sum_{i=1}^{n} \alpha(y - x_i) \right)^2 \right] - 2 f(y) \mathbb{W}_\alpha^\mathcal{P}(y, \mathbf{x}) + f(y)^2.
\]
To minimize \( \mathcal{F}(f, W_\alpha) \), we take the derivative of this expression with respect to \( f(y) \):
\[
\frac{d}{d f(y)} \mathcal{F}(f, W_\alpha) = -2 \mathbb{W}_\alpha^\mathcal{P}(y, \mathbf{x}) + 2 f(y).
\]
Setting this derivative to zero to find the minimum:
\[
-2 \mathbb{W}_\alpha^\mathcal{P}(y, \mathbf{x}) + 2 f(y) = 0,
\]
which gives:
\[
f(y) = \mathbb{W}_\alpha^\mathcal{P}(y, \mathbf{x}).
\]
Thus, the function that minimizes the expected \( \mathcal{F} \)-distance is \( f(y) = \mathbb{W}_\alpha^\mathcal{P}(y, \mathbf{x}) \).

This completes the proof that the expected welfare function \( \mathbb{W}_\alpha^\mathcal{P}(y, \mathbf{x}) \) is the best approximation of the empirical welfare function \( W_\alpha(y, \mathbf{x}) \) in the \( \mathcal{F} \)-distance sense.
\end{proof}

\subsection{Theorem 11}

\begin{theorem}[Asymptotic result for Welfare Functions]
Let \( W_\alpha(y, \mathbf{x}) \) be the empirical welfare function for a sample of \( n \) agents, where \( \mathbf{x} = (x_1, x_2, \dots, x_n) \) is a sequence of i.i.d. random variables sampled from the distribution \( \mathcal{P} \). As \( n \to \infty \), the empirical welfare converges in probability to the expected welfare for a single agent sampled from \( \mathcal{P} \):
\[
\frac{W_\alpha(y, \mathbf{x})}{n} \xrightarrow{p} \mathbb{W}_\alpha^\mathcal{P}(y, X),
\]
where \( X \sim \mathcal{P} \) and \( \mathbb{W}_\alpha^\mathcal{P}(y, X) \) is the expected welfare function for a single agent.
\end{theorem}

\begin{proof}
We begin by recalling the definition of the empirical welfare function \( W_\alpha(y, \mathbf{x}) \), which is the sum of the utilities for \( n \) agents located at \( \mathbf{x} = (x_1, x_2, \dots, x_n) \):
\[
W_\alpha(y, \mathbf{x}) = \sum_{i=1}^{n} \alpha(y - x_i),
\]
where \( \alpha(y - x_i) \) represents the utility of agent \( i \) given the facility location \( y \) and the agent’s location \( x_i \).

We seek to show that the normalized empirical welfare function \( \frac{W_\alpha(y, \mathbf{x})}{n} \) converges in probability to the expected welfare for a single agent, i.e.,
\[
\frac{W_\alpha(y, \mathbf{x})}{n} = \frac{1}{n} \sum_{i=1}^{n} \alpha(y - x_i) \xrightarrow{p} \mathbb{W}_\alpha^\mathcal{P}(y, X),
\]
where \( \mathbb{W}_\alpha^\mathcal{P}(y, X) = \mathbb{E}_{X \sim \mathcal{P}}[\alpha(y - X)] \).

We now apply the weak law of large numbers (WLLN), which states that for a sequence of i.i.d. random variables \( X_1, X_2, \dots, X_n \) with common distribution \( \mathcal{P} \), the sample average converges in probability to the expected value. Specifically, we apply WLLN to the sequence of random variables \( \alpha(y - x_i) \), where each \( x_i \) is drawn i.i.d. from \( \mathcal{P} \).

Since each \( x_i \sim \mathcal{P} \), the random variables \( \alpha(y - x_1), \alpha(y - x_2), \dots, \alpha(y - x_n) \) are i.i.d. with expected value:
\[
\mathbb{E}_{X \sim \mathcal{P}}[\alpha(y - X)] = \mathbb{W}_\alpha^\mathcal{P}(y, X).
\]
By the weak law of large numbers, the sample average of these random variables converges in probability to the expected value:
\[
\frac{1}{n} \sum_{i=1}^{n} \alpha(y - x_i) \xrightarrow{p} \mathbb{E}_{X \sim \mathcal{P}}[\alpha(y - X)] = \mathbb{W}_\alpha^\mathcal{P}(y, X).
\]

Thus, as \( n \to \infty \), the normalized empirical welfare function \( \frac{W_\alpha(y, \mathbf{x})}{n} \) converges in probability to the expected welfare function for a single agent:
\[
\frac{W_\alpha(y, \mathbf{x})}{n} \xrightarrow{p} \mathbb{W}_\alpha^\mathcal{P}(y, X).
\]
This completes the proof.
\end{proof}

\end{document}